\documentclass[12pt]{article}
\usepackage{a4,latexsym, amssymb}

\newtheorem{Theorem}{Theorem}
\newtheorem{Lemma}{Lemma}
\newtheorem{Corollary}{Corollary}

\newtheorem{Example}{Example}
\newtheorem{Remark}{Remark}

\begin{document}

\title{The Minimal Polynomial over $\mathbb{F}_q$ of
Linear Recurring Sequence over $\mathbb{F}_{q^m}$
\thanks{This research is supported in part by the National Natural Science Foundation of
China under the Grant 60872025. }}

\author{
Zhi-Han Gao\thanks{Z.-H. Gao is with the Chern Institute of
Mathematics, Nankai University, Tianjin 300071, P.R. China. E-mail: gaulwy@mail.nankai.edu.cn}
\ and Fang-Wei Fu\thanks{F.-W. Fu is with the Chern Institute of
Mathematics and the Key Laboratory of Pure Mathematics and
Combinatorics, Nankai University, Tianjin 300071, P.R. China. Email:
fwfu@nankai.edu.cn}}

\date{}
\maketitle

\begin{abstract}
Recently, motivated by the study of vectorized stream cipher
systems, the joint linear complexity and joint minimal polynomial of
multisequences have been investigated. Let $\mathcal{S}$ be a linear
recurring sequence over finite field $\mathbb{F}_{q^m}$ with minimal
polynomial $h(x)$ over $\mathbb{F}_{q^m}$. Since $\mathbb{F}_{q^m}$
and $\mathbb{F}_{q}^m$ are isomorphic vector spaces over the finite
field $\mathbb{F}_q$, $\mathcal{S}$ is identified with an $m$-fold
multisequence ${\bf S}^{(m)}$ over the finite field $\mathbb{F}_q$.
The joint minimal polynomial and joint linear complexity of the
$m$-fold multisequence ${\bf S}^{(m)}$ are the minimal polynomial
and linear complexity over $\mathbb{F}_q$ of $\mathcal{S}$
respectively. In this paper, we study the minimal polynomial and
linear complexity over $\mathbb{F}_q$ of a linear recurring sequence
$\mathcal{S}$ over $\mathbb{F}_{q^m}$ with minimal polynomial $h(x)$
over $\mathbb{F}_{q^m}$. If the canonical factorization of $h(x)$ in
$\mathbb{F}_{q^m}[x]$ is known, we determine the minimal polynomial
and linear complexity over $\mathbb{F}_q$ of the linear recurring
sequence $\mathcal{S}$ over $\mathbb{F}_{q^m}$.

\end{abstract}



\baselineskip=18pt

\section{Introduction}

Let $\mathbb{F}_{q^m}$ be a finite field with $q^m$ elements, which
contains a subfield $\mathbb{F}_q$ with $q$ elements. Let
$\mathcal{S}=(s_0,s_1,\ldots,s_n,\ldots)$ be a linear recurring
sequence over $\mathbb{F}_{q^m}$. The monic polynomial
$f(x)=a_0+a_1x+\cdots+a_{n-1}x^{n-1}+x^n \in \mathbb{F}_{q^m}[x]$ is
called a characteristic polynomial over $\mathbb{F}_{q^m}$ of
$\mathcal{S}$ if
\[
a_0s_k+a_1s_{k+1}+a_2s_{k+2}+\cdots+a_{n-1}s_{k+n-1}+s_{k+n}=0,\ \ \
\mbox{for all}\ k\geq 0. \] If the characteristic polynomial $f(x)$
is a polynomial over $\mathbb{F}_q$, that is, all $a_i \in
\mathbb{F}_q$, we call $f(x)$ a characteristic polynomial over
$\mathbb{F}_q$ of $\mathcal{S}$. Since the linear recurring sequence
$\mathcal{S}$ over $\mathbb{F}_{q^m}$ is ultimately periodic, a
characteristic polynomial over $\mathbb{F}_q$ of $\mathcal{S}$ does
exist. The minimal polynomial over $\mathbb{F}_{q^m}$ (resp.
$\mathbb{F}_q$) of $\mathcal{S}$ is the uniquely determined
characteristic polynomial over $\mathbb{F}_{q^m}$ (resp.
$\mathbb{F}_q$) of $\mathcal{S}$ with least degree. The linear
complexity over $\mathbb{F}_{q^m}$ (resp. $\mathbb{F}_q$) of
$\mathcal{S}$ is the degree of the minimal polynomial over
$\mathbb{F}_{q^m}$ (resp. $\mathbb{F}_q$) of $\mathcal{S}$. Let
$h(x)$ be the minimal polynomial over $\mathbb{F}_{q^m}$ of
$\mathcal{S}$. It is known that $h(x)|f(x)$ for any characteristic
polynomial $f(x)$ over $\mathbb{F}_{q^m}$ of $\mathcal{S}$.
Similarly, let $H(x)$ be the minimal polynomial over
$\mathbb{F}_{q}$ of $\mathcal{S}$, we have $H(x)|f(x)$ for any
characteristic polynomial $f(x)$ over $\mathbb{F}_{q}$ of
$\mathcal{S}$. Note that a characteristic polynomial $f(x)$ over
$\mathbb{F}_{q}$ of $\mathcal{S}$ is also a characteristic
polynomial over $\mathbb{F}_{q^m}$ of $\mathcal{S}$. Hence,
$h(x)|f(x)$ for any characteristic polynomial $f(x)$ over
$\mathbb{F}_{q}$ of $\mathcal{S}$. In particular, $h(x)|H(x)$.

Similarly, for any $m$-fold multisequence ${\bf
S}^{(m)}=(S_1,S_2,\ldots, S_m)$ over $\mathbb{F}_q$, the monic
polynomial $g(x)\in\mathbb{F}_q[x]$ is called a joint characteristic
polynomial of ${\bf S}^{(m)}$ if $g(x)$ is a characteristic
polynomial of $S_j$ for each $1\leq j\leq m$. The joint minimal
polynomial of ${\bf S}^{(m)}$ is the uniquely determined joint
characteristic polynomial of ${\bf S}^{(m)}$ with least degree, and
the joint linear complexity of ${\bf S}^{(m)}$ is the degree of the
joint minimal polynomial of ${\bf S}^{(m)}$. Since
$\mathbb{F}_{q^m}$ and $\mathbb{F}_{q}^m$ are isomorphic vector
spaces over the finite field $\mathbb{F}_q$, a linear recurring
sequence $\mathcal{S}$ over $\mathbb{F}_{q^m}$ is identified with an
$m$-fold multisequence ${\bf S}^{(m)}$ over $\mathbb{F}_q$. It is
well known that the joint minimal polynomial and joint linear
complexity of the $m$-fold multisequence ${\bf S}^{(m)}$ are the
minimal polynomial and linear complexity over $\mathbb{F}_q$ of
$\mathcal{S}$ respectively.

The linear complexity of sequences is one of the important security
measures for stream cipher systems (see \cite{cdr}, \cite{dxs},
\cite{ru1}, \cite{ru2}). For a general introduction to the theory of
linear feedback shift register sequences, we refer the reader to
\cite[Chapter 8]{ln} and the references therein. The linear
complexity of sequences has been extensively studied by many
researchers. For a recent survey paper, see Niederreiter \cite{n1}.
The notion of linear complexity over $\mathbb{F}_q$ of linear
recurring sequences over $\mathbb{F}_{q^m}$ was introduced by Ding,
Xiao and Shan in \cite{dxs}, and discussed by some authors, for
example, see \cite{cm}, \cite{kl}, \cite{mei}-\cite{mo2},
\cite{mus}, \cite{n1}, \cite{nv}. Recently, in the study of
vectorized stream cipher systems, the joint linear complexity of
multisequences has been extensively investigated (see \cite{diy},
\cite{ds}, \cite{fd}-\cite{hr}, \cite{mei}-\cite{nw2},
\cite{wn}-\cite{xi}).

In this paper, we study the minimal polynomial and linear complexity
over $\mathbb{F}_q$ of a linear recurring sequence $\mathcal{S}$
over $\mathbb{F}_{q^m}$ with minimal polynomial $h(x)$ over
$\mathbb{F}_{q^m}$. If the canonical factorization of $h(x)$ in
$\mathbb{F}_{q^m}[x]$ is known, we determine the minimal polynomial
and linear complexity over $\mathbb{F}_q$ of the linear recurring
sequence $\mathcal{S}$ over $\mathbb{F}_{q^m}$. The rest of the
paper is organized as follows. In Section \ref{lrs} we introduce and
give some results on linear recurring sequences that will be used in
this paper. In Section \ref{pra} we introduce a ring automorphism of
the polynomial ring $\mathbb{F}_{q^m}[x]$. We derive some results on
this polynomial ring automorphism that are crucial to establish the
main results in this paper. In Section \ref{mp} we determine the
minimal polynomial and linear complexity over $\mathbb{F}_q$ of a
linear recurring sequence $\mathcal{S}$ over $\mathbb{F}_{q^m}$ with
minimal polynomial $h(x)$ over $\mathbb{F}_{q^m}$. In Section
\ref{lbmo} we give a new proof for the lower bound of Meidl and
\"Ozbudak \cite{mo1} on the linear complexity over
$\mathbb{F}_{q^m}$ of linear recurring sequence $\mathcal{S}$ over
$\mathbb{F}_{q^m}$ with given minimal polynomial $g(x)$ over
$\mathbb{F}_q$. We show that this lower bound is tight if and only
if the minimal polynomial over $\mathbb{F}_{q^m}$ of $\mathcal{S}$
is in certain form.

\section{Linear Recurring Sequences}
\label{lrs}

Let $f(x)$ be a monic polynomial over $\mathbb{F}_q$. Denote
$\mathcal {M}(f(x))$ the set of all linear recurring sequences over
$\mathbb{F}_q$ with characteristic polynomial $f(x)$. Note that
$\mathcal {M}(f(x))$ is a vector space over $\mathbb{F}_q$ with
dimension $\mbox{deg}(f(x))$.  We need the following results on
linear recurring sequences from \cite{ln}:
\begin{Theorem}
\label{th1} {\rm \cite[Theorem 8.55]{ln}}\quad Let $f_1(x),\ldots,
f_k(x)$ be monic polynomials over $\mathbb{F}_q$. If $f_1(x),\ldots,
f_k(x)$ are pairwise relatively prime, then the vector space\\
$\mathcal{M}(f_1(x)\cdots f_k(x))$ is the direct sum of the
subspaces $\mathcal{M}(f_1(x)),\cdots, \mathcal{M}(f_k(x))$, that is
\[
\mathcal{M}(f_1(x)\cdots f_k(x))=\mathcal{M}(f_1(x))\dotplus \cdots
\dotplus \mathcal{M}(f_k(x)).
\]
\end{Theorem}

\begin{Theorem}
\label{th2} {\rm \cite[Theorem 8.57]{ln}}\quad Let $S_1, S_2,
\ldots, S_k$ be linear recurring sequences over $\mathbb{F}_{q}$.
The minimal polynomials over $\mathbb{F}_{q}$ of $S_1, S_2, \ldots,
S_k$ are $h_1(x), h_2(x), \ldots, h_k(x)$ respectively. If $h_1(x),
h_2(x), \ldots, h_k(x)$ are pairwise relatively prime, then the
minimal polynomial over $\mathbb{F}_{q}$ of $\sum_{i=1}^{k}S_i$ is
the product of $h_1(x),h_2(x),\ldots, h_k(x)$.
\end{Theorem}
It is easy to extend this result to the following case:
\begin{Lemma}
\label{lemma1} Let $\mathcal{S}_1, \mathcal{S}_2, \ldots,
\mathcal{S}_k$ be linear recurring sequences over
$\mathbb{F}_{q^m}$. The minimal polynomials over $\mathbb{F}_{q}$ of
$\mathcal{S}_1, \mathcal{S}_2, \ldots, \mathcal{S}_k$ are $H_1(x),
H_2(x), \ldots, H_k(x)$ respectively. If $H_1(x), H_2(x), \ldots,
H_k(x)$ are pairwise relatively prime over $\mathbb{F}_{q}$, then
the minimal polynomial over $\mathbb{F}_{q}$ of
$\sum_{i=1}^{k}\mathcal{S}_i$ is the product of
$H_1(x),H_2(x),\ldots, H_k(x)$.
\end{Lemma}
Now we establish the following lemma which will be used in this paper:
\begin{Lemma}
\label{lemma2} Let $S$ be a linear recurring sequence over
$\mathbb{F}_{q}$. The minimal polynomial over $\mathbb{F}_{q}$ of
$S$ is given by $h(x)=h_1(x)h_2(x)\cdots h_k(x)$ where $h_1(x),
h_2(x), \ldots, h_k(x)$ are monic polynomials over $\mathbb{F}_{q}$.
If $h_1(x), h_2(x), \ldots, h_k(x)$ are pairwise relatively prime,
then there uniquely exist sequences $S_1, S_2, \ldots, S_k$ over
$\mathbb{F}_{q}$ such that
\[
S=S_1+S_2+\cdots+S_k
\]
and the minimal polynomials over $\mathbb{F}_{q}$ of
$S_1,S_2,\ldots,S_k$ are $h_1(x), h_2(x), \ldots, h_k(x)$
respectively.
\end{Lemma}
\begin{proof}
By Theorem \ref{th1}, we have
\[
\mathcal{M}(h(x))=\mathcal{M}(h_1(x))\dotplus \cdots \dotplus
\mathcal{M}(h_k(x)).
\]
Then, there uniquely exist sequences $S_1,S_2,\ldots,S_k$ over
$\mathbb{F}_{q}$ such that $S_j\in \mathcal{M}(h_j(x))$ and
\[
S=S_1+S_2+\cdots+S_k.
\]
Assume that the minimal polynomial over $\mathbb{F}_{q}$ of $S_j$ is
$h_j^{'}(x)$ which is a divisor of $h_j(x)$ for $1\leq j\leq k$. By
Theorem \ref{th2}, the minimal polynomial over $\mathbb{F}_{q}$ of
$S$ is $\prod_{j=1}^{k}h_j^{'}(x)$. Thus,
\[
h_1^{'}(x)h_2^{'}(x)\cdots h_k^{'}(x)=h_1(x)h_2(x)\cdots h_k(x).
\]
Since $h_j^{'}(x)| h_j(x)$ for $1\leq j\leq k$, we have
\[
h_j^{'}(x)=h_j(x), \;\; 1\leq j\leq k,
\]
which completes the proof.
\end{proof}

\section{Polynomial Ring Automorphism}
\label{pra}

We define $\sigma$ to be a mapping from the polynomial ring
$\mathbb{F}_{q^m}[x]$ to itself as follows: For
$f(x)=a_0+a_1x+\cdots+a_nx^n\in \mathbb{F}_{q^m}[x]$,
\[\sigma: \mathbb{F}_{q^m}[x]\longrightarrow \mathbb{F}_{q^m}[x],\]
\[f(x)\longrightarrow \sigma(f(x))\]
where $\sigma(f(x))=a_0^q+a_1^qx+\cdots+a_n^qx^n$. It is easy to see
that $\sigma$ is a ring automorphism of $\mathbb{F}_{q^m}[x]$.
Throughout the paper, we will use the fact that
\[ \sigma(f(x)g(x))=\sigma(f(x))\sigma(g(x)), \;\;\mbox{for any} \; f(x), g(x)\in
\mathbb{F}_{q^m}[x]. \] Denote $\sigma^{(k)}$ the $k$th usual
composition of $\sigma$. Note that $\sigma^{(0)}$ is the identity
mapping. Since $a^{q^m}=a$ for any $a\in \mathbb{F}_{q^m}$, we have
$\sigma^{(m)}(f(x))=f(x)$. Denote $k(f)$ the minimum positive
integer $k$ such that $\sigma^{(k)}(f(x))=f(x)$.
\begin{Lemma}
\label{lemma3}
For any $f(x)\in \mathbb{F}_{q^m}[x]$ and positive integer $l$,
$\sigma^{(l)}(f(x))=f(x)$ if and only if $k(f)|l$.
\end{Lemma}
\begin{proof}
It is easy to see that $\sigma^{(l)}(f(x))=f(x)$ if $k(f)|l$.
On the other hand, if $\sigma^{(l)}(f(x))=f(x)$, we assume that
$l=k(f)w+r$ and $0\leq r<k(f)$. Then
\[ f(x)=\sigma^{(l)}(f(x))=\sigma^{(r)}(\sigma^{(k(f)w)}(f(x)))=\sigma^{(r)}(f(x)). \]
Hence, $r=0$ by the definition of $k(f)$. Therefore, $k(f)|l$.
\end{proof}

Now we define an equivalence relation $\stackrel{\sigma}{\sim}$ on
$\mathbb{F}_{q^m}[x]$: $f(x)\stackrel{\sigma}{\sim} g(x)$ if and
only if there exists positive integer $j$ such that
$\sigma^{(j)}(f(x))=g(x)$. The equivalence classes induced by this
equivalence relation $\stackrel{\sigma}{\sim}$ are called
$\sigma$-equivalence classes.

\begin{Lemma}
\label{lemma4} Let $f(x)$ be a polynomial over $\mathbb{F}_{q^m}$.
Then $\sigma(f(x))$ is irreducible over $\mathbb{F}_{q^m}$ if and
only if $f(x)$ is irreducible over $\mathbb{F}_{q^m}$.
\end{Lemma}
\begin{proof}
Since $f(x)\in\mathbb{F}_{q^m}[x]$, we have
$f(x)=\sigma^{(m)}(f(x))$. Then, we only need to prove that
$\sigma(f(x))$ is irreducible over $\mathbb{F}_{q^m}$ if $f(x)$ is
irreducible over $\mathbb{F}_{q^m}$. Assume that $\sigma(f(x))$ is
not irreducible over $\mathbb{F}_{q^m}$, that is to say there exist
two nonconstant polynomials $r_1(x),r_2(x)$ in $\mathbb{F}_{q^m}[x]$
such that $\sigma(f(x))=r_1(x)r_2(x)$. Therefore,
\[f(x)=\sigma^{(m)}(f(x))=\sigma^{(m-1)}(\sigma(f(x)))=\sigma^{(m-1)}(r_1(x))\sigma^{(m-1)}(r_2(x))\]
where $\sigma^{(m-1)}(r_1(x)), \sigma^{(m-1)}(r_2(x))$ are
nonconstant polynomials over $\mathbb{F}_{q^m}$, which contradicts
to the fact that $f(x)$ is irreducible over $\mathbb{F}_{q^m}$.
Hence, $\sigma(f(x))$ is irreducible over $\mathbb{F}_{q^m}$.
\end{proof}
The following theorem is crucial to establish the main results in
this paper.
\begin{Theorem}
\label{th3} Let $f(x)$ be an irreducible polynomial in
$\mathbb{F}_{q^m}[x]$, then the product
\[
f(x)\sigma(f(x))\sigma^{(2)}(f(x))\cdots \sigma^{(k(f)-1)}(f(x))
\]
is an irreducible polynomial in $\mathbb{F}_{q}[x]$.
\end{Theorem}
\begin{proof}
Let $\mbox{deg}(f(x))=n$. Then, by \cite[Chapter 2, Theorem
2.14]{ln} there exits $\alpha\in \mathbb{F}_{q^{mn}}$ such that
\begin{eqnarray}
\label{pr1}
f(x)=(x-\alpha)(x-\alpha^{q^m})(x-\alpha^{q^{2m}})\cdots(x-\alpha^{q^{(n-1)m}})
\end{eqnarray}
where $\alpha,\alpha^{q^m},\ldots,\alpha^{q^{(n-1)m}}$ are different
roots of $f(x)$. Let $g(x)$ be the minimal polynomial of
$\alpha\in\mathbb{F}_{q^{mn}}$ over $\mathbb{F}_q$. By \cite[Chapter
2, Theorem 2.14]{ln}, $g(x)$ is an irreducible polynomial over
$\mathbb{F}_q$ and
\begin{eqnarray}
\label{pr2}
g(x)=(x-\alpha)(x-\alpha^{q})(x-\alpha^{q^2})\cdots(x-\alpha^{q^{d-1}})
\end{eqnarray}
where $d$ is the least positive integer such that
$\alpha^{q^d}=\alpha$. Since $\alpha^{q^{mn}}=\alpha$ and
$\alpha,\alpha^{q^m},\ldots,\alpha^{q^{(n-1)m}}$ are distinct, we
have $d\mid mn$ but $d\nmid im$ for $1\leq i\leq n-1$. Then, we
claim that $d$ must be a multiple of $n$. Otherwise, we have $
\gcd(d,n)<n$. Since $d\mid mn$, then we have $\frac{d}{
\gcd(d,n)}\mid\frac{mn}{\gcd(d,n)}$. Since $\frac{d}{ \gcd(d,n)}$
and $\frac{n}{ \gcd(d,n)}$ are relatively prime, we have $\frac{d}{
\gcd(d,n)}\mid m$. Then, $d\mid \gcd(d,n)m$. This gives a
contradiction since $\gcd(d,n)<n$. Therefore, $d$ is a multiple of
$n$. Let $k$ be the positive integer such that $d=nk$. Since $d\mid
mn$, then $k\mid m$. Let $s$ be the positive integer such that
$m=sk$. Then, we claim that $s$ and $n$ are relatively prime.
Otherwise, we have $\frac{n}{\gcd(n,s)}<n$. Since $n\mid
\frac{ns}{\gcd(n,s)}$, then $kn\mid \frac{kns}{ \gcd(n,s)}$, that is
$d\mid m\frac{n}{\gcd(n,s)}$. This gives a contradiction since
$\frac{n}{\gcd(n,s)}<n$. Therefore, $s$ and $n$ are relatively
prime. Thus, $\{js | j=0,1,\ldots, n-1\}$ is a complete residue
system modulo $n$, i.e., there exits $(i_0,i_1,\ldots,i_{n-1})$, a
permutation of $(0,1,2,\ldots,n-1)$, such that $js\equiv i_j
\;\;({\rm mod} \ n)$. So we have $kjs\equiv ki_j \;\;({\rm mod} \
kn)$, i.e., $jm\equiv ki_j \;\;({\rm mod} \ d)$. Hence,
$\alpha^{q^{jm}}=\alpha^{q^{ki_j}}$ for $0\leq j\leq n-1$.
Therefore, it follows from (\ref{pr1}) that
\begin{eqnarray}
f(x)&=&(x-\alpha^{q^{ki_{0}}})(x-\alpha^{q^{ki_{1}}})(x-\alpha^{q^{ki_{2}}})\cdots
(x-\alpha^{q^{ki_{n-1}}})\nonumber\\
&=&(x-\alpha)(x-\alpha^{q^{k}})(x-\alpha^{q^{2k}})\cdots(x-\alpha^{q^{(n-1)k}}).
\label{pr3}
\end{eqnarray}
By (\ref{pr3}) and the definition of $\sigma$, we have
\begin{eqnarray}
\label{pr4}
\sigma^{(i)}(f(x))=(x-\alpha^{q^i})(x-\alpha^{q^{k+i}})(x-\alpha^{q^{2k+i}})\cdots(x-\alpha^{q^{(n-1)k+i}}).
\end{eqnarray}
By (\ref{pr2}), (\ref{pr3}), (\ref{pr4}) and note that $d=nk$, we
have
\[ g(x)=f(x)\sigma(f(x))\ldots\sigma^{(k-1)}(f(x)) \] and
\[
\sigma^{(k)}(f(x))=(x-\alpha^{q^k})(x-\alpha^{q^{2k}})(x-\alpha^{q^{3k}})\ldots(x-\alpha^{q^{nk}})=f(x).
\]
Since $d$ is the least positive integer such that
$\alpha^{q^d}=\alpha$ and $d=nk$, we have that
$f(x),\sigma(f(x)),\ldots,\sigma^{(k-1)}(f(x))$ are different from
each other. Hence, $k=k(f)$. Therefore,
\[
g(x)=f(x)\sigma(f(x))\cdots\sigma^{(k(f)-1)}(f(x)).
\]
Note that $g(x)$ is an irreducible polynomial over $\mathbb{F}_q$,
we complete the proof.
\end{proof}

Let $f(x)$ be an irreducible polynomial in $\mathbb{F}_{q^m}[x]$. It
is known from Lemma \ref{lemma4} that $f(x), \sigma(f(x)), \ldots,
\sigma^{(k(f)-1)}(f(x))$ are irreducible polynomials in
$\mathbb{F}_{q^m}[x]$. Denote
\[ R(f(x))=f(x)\sigma(f(x))\cdots\sigma^{(k(f)-1)}(f(x)). \]
By Theorem \ref{th3}, $R(f(x))$ is irreducible in
$\mathbb{F}_{q}[x]$. Note that $R(f(x))$ is a multiple of $f(x)$ in
$\mathbb{F}_{q^m}[x]$. Using Theorem \ref{th3}, we could give an
refined version of \cite[Chapter 3, Theorem
 3.46]{ln} as follows:
\begin{Theorem}
\label{th4} Let $f(x)$ be a monic irreducible polynomial over
$\mathbb{F}_{q}$ and $n=\deg(f(x))$. Let $m$ be a positive integer.
Denote $u= {\gcd}(n,m)$. Then the canonical factorization of $f(x)$
into monic irreducibles over $\mathbb{F}_{q^m}$ is given by
\[
f(x)=h(x)\sigma(h(x))\cdots \sigma^{(k(h)-1)}(h(x))
\]
where $h(x)$ is a monic irreducible polynomial over
$\mathbb{F}_{q^m}$ and $k(h)=u$.
\end{Theorem}
\begin{proof}
By \cite[Chapter 3, Theorem
 3.46]{ln}, the canonical factorization of $f(x)$
into monic irreducibles over $\mathbb{F}_{q^m}$ is given by
\[
f(x)=f_1(x)f_2(x)\cdots f_u(x)
\] where $f_1(x),f_2(x),\ldots,f_u(x)\in \mathbb{F}_{q^m}[x]$ are
distinct irreducible polynomials with the same degree. Let
$h(x)=f_1(x)$. By Theorem \ref{th3}, $R(h(x))$ is an irreducible
polynomial in $\mathbb{F}_{q}[x]$. Since $f(x)$ and $R(h(x))$ have a
common factor $h(x)$ in $\mathbb{F}_{q^m}[x]$, $f(x)$ and $R(h(x))$
are not relatively prime in $\mathbb{F}_{q}[x]$. Note that $f(x)$
and $R(h(x))$ are monic irreducible polynomials in
$\mathbb{F}_{q}[x]$. So, $f(x)=R(h(x))$. By Lemma \ref{lemma4},
$h(x),\sigma(h(x)),\ldots ,\sigma^{(k(h)-1)}(h(x))$ are all
irreducible polynomials over $\mathbb{F}_{q^m}$. Therefore, the
canonical factorization of $f(x)$ into monic irreducibles over
$\mathbb{F}_{q^m}$ is given by
\[
f(x)=h(x)\sigma(h(x))\cdots \sigma^{(k(h)-1)}(h(x))
\] and $k(h)=u$.
\end{proof}

In certain sense, Theorem \ref{th4} could be considered as a converse
procedure of Theorem \ref{th3}.

\section{Minimal Polynomials over $\mathbb{F}_{q}$ and $\mathbb{F}_{q^m}$}
\label{mp}

Now we determine the minimal polynomial and linear complexity over
$\mathbb{F}_q$ of a linear recurring sequence $\mathcal{S}$ over
$\mathbb{F}_{q^m}$ with minimal polynomial $h(x)\in
\mathbb{F}_{q^m}[x]$.
\begin{Theorem}
\label{th5} Let $\mathcal{S}$ be a linear recurring sequence over
$\mathbb{F}_{q^m}$ with minimal polynomial $h(x)\in
\mathbb{F}_{q^m}[x]$. Assume that the canonical factorization of
$h(x)$ in $\mathbb{F}_{q^m}[x]$ is given by
\[
h(x)=\prod_{j=1}^{l}P_{j0}^{e_{j0}}P_{j1}^{e_{j1}}\cdots
P_{ji_j}^{e_{ji_j}}
\]
where $\{P_{uv}\}$ are distinct monic irreducible polynomials in
$\mathbb{F}_{q^m}[x]$, $P_{j0},P_{j1},\ldots, P_{ji_j}$ are in the
same $\sigma$-equivalence class and $P_{uv}$, $P_{tw}$ are in the
different $\sigma$-equivalence classes when $u\neq t$. Then the
minimal polynomial over $\mathbb{F}_q$ of $\mathcal{S}$ is given by
\[
H(x)=\prod_{j=1}^{l}R(P_{j0})^{e_j}
\]
where $e_j=\max\{e_{j0},e_{j1},\ldots,e_{ji_j}\}$ for $1\leq j\leq
l$.
\end{Theorem}
\begin{proof}
By Lemma \ref{lemma2}, there uniquely exist sequences
$\mathcal{S}_1,\mathcal{S}_2,\ldots,\mathcal{S}_l$ over
$\mathbb{F}_{q^m}$ such that
\[
\mathcal{S}=\mathcal{S}_1+\mathcal{S}_2+\cdots+\mathcal{S}_l
\] and the minimal polynomial over $\mathbb{F}_{q^m}$ of $\mathcal{S}_j$ is $P_{j0}^{e_{j0}}P_{j1}^{e_{j1}}\cdots
P_{ji_j}^{e_{ji_j}}$ for $1\leq j\leq l$. Let $H_j(x)$ be the
minimal polynomial over $\mathbb{F}_{q}$ of $\mathcal{S}_j$. Since
$P_{j0},P_{j1},\ldots ,P_{ji_j}$ are in the same
$\sigma$-equivalence class, then $R(P_{j0})^{e_j}$ is a multiple of
$P_{j0}^{e_{j0}}P_{j1}^{e_{j1}}\cdots P_{ji_j}^{e_{ji_j}}$. So, by
Theorem \ref{th3}, $R(P_{j0})^{e_j}$ is a characteristic polynomial
over $\mathbb{F}_q$ of $\mathcal{S}_j$. Hence, $H_j(x)$ divides
$R(P_{j0})^{e_j}$ in $\mathbb{F}_q[x]$. Since, by Theorem \ref{th3},
$R(P_{j0})$ is irreducible over $\mathbb{F}_q$, we have
$H_j(x)=R(P_{j0})^{e'_j}$ where $e'_j\leq e_j$. By the definition of
$e_j$, there exists $e_{ju_j}$ such that $e_{ju_j}=e_j$ where $0\leq
u_j\leq i_j$. If $e'_j<e_j$, then $P_{ju_j}^{e_{ju_j}}$ can't divide
$H_j(x)$. However, $H_j(x)$ is a multiple of
$P_{j0}^{e_{j0}}P_{j1}^{e_{j1}}\cdots P_{ji_j}^{e_{ji_j}}$ in
$\mathbb{F}_{q^m}[x]$ since $H_j(x)$ is also a characteristic
polynomial over $\mathbb{F}_{q^m}$ of $\mathcal{S}_j$. This gives a
contradiction. Therefore, $e'_j=e_j$, i.e.,
$H_j(x)=R(P_{j0})^{e_j}$. For any $0\leq u\neq v\leq l$, we claim
that $R(P_{u0})^{e_u}$ and $R(P_{v0})^{e_v}$ are relatively prime.
Suppose on the contrary that there exist $R(P_{u0})^{e_u}$ and
$R(P_{v0})^{e_v}$, where $u\neq v$, which are not relatively prime.
Since $R(P_{u0})$ and $R(P_{v0})$ are monic irreducible polynomials
over $\mathbb{F}_q$, then we have $R(P_{u0})=R(P_{v0})$. Hence,
$P_{u0}$ divides $R(P_{v0})$ in $\mathbb{F}_{q^m}[x]$. By Theorem
\ref{th4}, the canonical factorization of $R(P_{v0})$ in
$\mathbb{F}_{q^m}[x]$ is given by
\[
R(P_{v0})=P_{v0}\sigma(P_{v0})\cdots\sigma^{(k(P_{v0})-1)}(P_{v0}).
\]
Since $P_{u0}$ is irreducible over $\mathbb{F}_{q^m}$, there exists
a positive integer $j$ such that $P_{u0}=\sigma^{(j)}(P_{v0})$. This
contradicts to the fact that $P_{u0}$ and $P_{v0}$ are in the
different $\sigma$-equivalence classes. Therefore, $R(P_{u0})^{e_u}$
and $R(P_{v0})^{e_v}$ are relatively prime. Then,
$H_1(x),H_2(x),\ldots,H_l(x)$ are pairwise relatively prime. By
Lemma \ref{lemma1}, the minimal polynomial over $\mathbb{F}_{q}$ of
$\mathcal{S}=\sum_{j=1}^{l}\mathcal{S}_j$ is the product of
$H_1(x),H_2(x),\ldots, H_l(x)$. Therefore, we have
\[
H(x)=\prod_{j=1}^{l}R(P_{j0})^{e_j}
\] which completes the proof.
\end{proof}

\begin{Corollary}
\label{cor1} Under the notation of Theorem \ref{th5}, the linear
complexity over $\mathbb{F}_q$ of $\mathcal{S}$ is given by
\[
L_{\mathbb{F}_q}(\mathcal{S})=\sum_{j=1}^{l}e_jk(P_{j0})\deg(P_{j0})
\] where $k(f)$ is defined in Section \ref{pra}.
\end{Corollary}

Using Theorem \ref{th5}, we could also give a refinement of
\cite[Proposition 2.1]{mo2}:
\begin{Theorem}
\label{th6} Let $f(x)$ be a polynomial over $\mathbb{F}_{q}$ with
$\deg(f)\geq 1$. Suppose that
\begin{eqnarray}
\label{mf1} f=r_1^{e_1}r_2^{e_2}\cdots r_l^{e_l},\mbox{~~~}e_1, e_2,
\ldots, e_l>0
\end{eqnarray}
is the canonical factorization of $f$ into monic irreducibles over
$\mathbb{F}_{q}$. Denote $n_i=\deg(r_i)$. Suppose by Theorem
\ref{th4} that the canonical factorization of $r_i(x)$ into monic
irreducibles over $\mathbb{F}_{q^m}$ is given by
\begin{eqnarray}
\label{mf2} r_i(x)=P_{i}(x)\sigma^{(1)}(P_{i}(x))\cdots
\sigma^{(u_i-1)}(P_{i}(x))
\end{eqnarray}
where $u_i= \gcd(n_i,m)=k(P_{i}(x))$. Let $\mathcal{S}$ be a linear
recurring sequence over $\mathbb{F}_{q^m}$. Then, the minimal
polynomial over $\mathbb{F}_{q}$ of $\mathcal{S}$ is $f(x)$ if and
only if the minimal polynomial $h(x)$ over $\mathbb{F}_{q^m}$ of
$\mathcal{S}$ is of the following form:
\begin{equation}
\label{mf3}
h(x)=\prod_{i=1}^{l}P_{i}^{e_{i0}}\sigma^{(1)}(P_{i})^{e_{i1}}\cdots
\sigma^{({u_i-1})}(P_{i})^{e_{iu_i-1}}
\end{equation}
where $0\leq e_{ij}\leq e_i$ and
$\max\{e_{i0},e_{i1},\ldots,e_{iu_i-1}\}=e_i$ for every
$i=1,2,\ldots,l$.
\end{Theorem}
\begin{proof}
It follows from Theorem \ref{th5} that the minimal polynomial over
$\mathbb{F}_{q}$ of $\mathcal{S}$ is $f(x)$ if the minimal
polynomial $h(x)$ over $\mathbb{F}_{q^m}$ of $\mathcal{S}$ is given
by (\ref{mf3}).

Conversely, suppose that the minimal polynomials over
$\mathbb{F}_{q}$ of $\mathcal{S}$ is $f(x)$. Then, $h(x)$ is a
factor of $f(x)$ in $\mathbb{F}_{q^m}[x]$ since $f(x)$ is also a
characteristic polynomial over $\mathbb{F}_{q^m}$ of $\mathcal{S}$.
By (\ref{mf1}) and (\ref{mf2}), the canonical factorization of
$f(x)$ into monic irreducibles over $\mathbb{F}_{q^m}$ is given by
\[
f(x)=\prod_{i=1}^{l}P_{i}^{e_{i}}\sigma^{(1)}(P_{i})^{e_{i}}\cdots
\sigma^{({u_i-1})}(P_{i})^{e_{i}}.
\]
So $h(x)$ must be of the form
\[
h(x)=\prod_{i=1}^{l}P_{i}^{e_{i0}}\sigma^{(1)}(P_{i})^{e_{i1}}\cdots
\sigma^{({u_i-1})}(P_{i})^{e_{iu_i-1}}
\]
where $0\leq e_{ij}\leq e_i$ for every $i=1,2,\ldots,l$. By Theorem
\ref{th5}, the minimal polynomial over $\mathbb{F}_{q}$ of
$\mathcal{S}$ is given by
\[ H(x)=\prod_{i=1}^{l}R(P_{i})^{e'_i}=\prod_{i=1}^{l}r_i(x)^{e'_i}
\]
where $e'_i=\max\{e_{i0},e_{i1},\ldots,e_{iu_i-1}\}$. Due to the
uniqueness of the minimal polynomial over $\mathbb{F}_{q}$ of
$\mathcal{S}$, we have $H(x)=f(x)$. Hence, $e'_i=e_i$. Therefore,
the minimal polynomial $h(x)$ over $\mathbb{F}_{q^m}$ of
$\mathcal{S}$ is of the form (\ref{mf3}). This completes the proof.
\end{proof}

At the end of this section, we give an example to illustrate Theorem \ref{th5} and Corollary \ref{cor1}.
\begin{Example}
Let $\mathbb{F}_2\subseteq \mathbb{F}_4$ and let $\alpha$ be a root
of $x^2+x+1$ in $\mathbb{F}_4$. So, $\mathbb{F}_4 =\{0,1, \alpha,
1+\alpha \}$. Let $\mathcal{S}$ be a periodic sequence over
$\mathbb{F}_4$ with the least period $15$. The first period terms of
$\mathcal{S}$ are given by
\[
\alpha^2,\alpha,\alpha,\alpha^2,\alpha^2,\alpha^2,0,\alpha,\alpha^2,\alpha,0,\alpha,0,0,1.
\]
The minimal polynomial over $\mathbb{F}_{4}$ of $\mathcal{S}$ is
$x^3+\alpha^2x^2+\alpha^2$. We first factor
$x^3+\alpha^2x^2+\alpha^2$ into irreducible polynomials over
$\mathbb{F}_4$:
\[
x^3+\alpha^2x^2+\alpha^2=(x+\alpha)(x^2+x+\alpha).
\]
Note that
\[
\sigma(x+\alpha)=x+\alpha^2,\;\; \sigma^{(2)}(x+\alpha)=x+\alpha,
\]
\[
\sigma(x^2+x+\alpha)=x^2+x+\alpha^2, \;\;
\sigma^{(2)}(x^2+x+\alpha)=x^2+x+\alpha.
\]
So we have
\[k(x+\alpha)=2, \;\; k(x^2+x+\alpha)=2.\]
Then, by Theorem \ref{th5} and Corollary \ref{cor1}, the minimal polynomial over $\mathbb{F}_{2}$ of $\mathcal{S}$ is
\begin{eqnarray}
&& (x+\alpha)\sigma(x+\alpha)(x^2+x+\alpha)\sigma(x^2+x+\alpha) \\
&=&(x^2+x+1)(x^4+x+1)=x^6+x^5+x^4+x^3+1
\end{eqnarray}
and the linear complexity over $\mathbb{F}_{2}$ of $\mathcal{S}$ is
\[
L=1\times k(x+\alpha)\times \deg(x+\alpha)+1\times
k(x^2+x+\alpha)\times \deg(x^2+x+\alpha)=2+2\times2=6.
\]
\end{Example}

\section{Remarks on the Lower Bound of
Meidl and \"Ozbudak }
\label{lbmo}

Meidl and \"Ozbudak \cite{mo1} derived a lower bound on the linear
complexity over $\mathbb{F}_{q^m}$ of a linear recurring sequence
$\mathcal{S}$ over $\mathbb{F}_{q^m}$ with given minimal polynomial
$g(x)$ over $\mathbb{F}_q$. In this section, using Theorem \ref{th6}
we give a new proof for the lower bound of Meidl and \"Ozbudak and
show that this lower bound is tight if and only if the minimal
polynomial over $\mathbb{F}_{q^m}$ of $\mathcal{S}$ is in certain
form.

\begin{Corollary}
\label{cor2} Let $f(x)$ be a monic polynomial in $\mathbb{F}_{q}[x]$
with the canonical factorization into irreducible polynomials over
$\mathbb{F}_{q}$ given by
\begin{eqnarray}
\label{lb1} f=r_1^{e_1}r_2^{e_2}\ldots r_k^{e_k}, \;\;\; e_1, e_2,
\ldots, e_k>0.
\end{eqnarray}
Suppose that $\mathcal{S}$ is a linear recurring sequence over
$\mathbb{F}_{q^m}$ and the minimal polynomial over $\mathbb{F}_{q}$
of $\mathcal{S}$ is $f(x)$. Then, the linear complexity
$L_{\mathbb{F}_{q^m}}(\mathcal{S})$ over $\mathbb{F}_{q^m}$ of
$\mathcal{S}$ is lower bounded by
\begin{eqnarray}
\label{lb2} L_{\mathbb{F}_{q^m}}(\mathcal{S})\geq
 \sum_{i=1}^{k}e_i\frac{n_i}{\gcd(n_i,m)}
\end{eqnarray}
where $n_i=\deg(r_i)$ for $i=1,2,\ldots, k$. Furthermore, suppose by
Theorem \ref{th4} that the canonical factorization of $r_i(x)$ into
monic irreducibles over $\mathbb{F}_{q^m}$ is given by
\begin{eqnarray}
\label{lb3} r_i(x)=P_{i}(x)\sigma^{(1)}(P_{i}(x))\ldots
\sigma^{(u_i-1)}(P_{i}(x))
\end{eqnarray}
where $u_i= \gcd(n_i,m)$ for $i=1,2,\ldots, k$. Then, the lower
bound is tight if and only if the minimal polynomial $h(x)$ over
$\mathbb{F}_{q^m}$ of $\mathcal{S}$ is of the following form:
\[ h(x)=
 \prod _{i=1}^{k} \sigma^{(j_i)}(P_{i})^{e_i}
\]
where $0\leq j_i\leq u_i -1$ for $i=1,2,\ldots, k$.
\end{Corollary}
\begin{proof}
It follows from (\ref{lb1}) and (\ref{lb3}) and Theorem \ref{th6}
that the minimal polynomial $h(x)$ over $\mathbb{F}_{q^m}$ of
$\mathcal{S}$ is of the form:
\begin{equation}
\label{lb4}
h(x)=\prod_{i=1}^{k}P_{i}^{e_{i0}}\sigma^{(1)}(P_{i})^{e_{i1}}\cdots
\sigma^{({u_i-1})}(P_{i})^{e_{iu_i-1}}
\end{equation}
where $0\leq e_{ij}\leq e_i$ and
$\max\{e_{i0},e_{i1},\ldots,e_{iu_i-1}\}=e_i$ for every
$i=1,2,\ldots,k$. Note from (\ref{lb3}) that
$\deg(P_{i}(x))=n_i/u_i$. Hence, by (\ref{lb4}),
\begin{eqnarray*}
L_{\mathbb{F}_{q^m}}(\mathcal{S})=\deg(h(x)) \geq \sum_{i=1}^{k}e_i
\deg(P_{i}(x)) =
 \sum_{i=1}^{k}e_i\frac{n_i}{ \mbox{gcd}(n_i,m)}
\end{eqnarray*}
and the equality holds if and only if
\[ h(x)=
 \prod _{i=1}^{k} \sigma^{(j_i)}(P_{i})^{e_i}
\]
where $0\leq j_i\leq u_i -1$ for $i=1,2,\ldots, k$. This completes
the proof.
\end{proof}

\begin{Remark}
Meidl and \"Ozbudak \cite[Proposition 3]{mo1} showed that there
exists a linear recurring sequence over $\mathbb{F}_{q^m}$ such that
the lower bound (\ref{lb2}) is tight. We give in Corollary
\ref{cor2} the necessary and sufficient condition under which the
lower bound (\ref{lb2}) is tight.
\end{Remark}


\baselineskip=14pt\small

\begin{thebibliography}{99}

\bibitem{cm} W.-S. Chou, G.L. Mullen, Generating linear spans over finite fields, Acta
Arith. 61 (1992), 183-191.

\bibitem{cdr} T.W. Cusick, C. Ding, A. Renvall, Stream Ciphers and Number Theory,
Elsevier, Amsterdam, 1998.

\bibitem{diy} Z. Dai, K. Imamura, J. Yang, Asymptotic behavior of normalized linear complexity
of multi-sequences, in: T. Helleseth et al. (Eds.), Sequences and Their
Applications --- SETA 2004, Lecture Notes in Computer Science, Vol.
3486, Springer, Berlin, 2005, pp. 129--142.

\bibitem{ds}
E. Dawson, L. Simpson, Analysis and design issues for synchronous
stream ciphers, in: H. Niederreiter (Ed.), Coding Theory and Cryptology,
World Scientific, Singapore, 2002, pp. 49--90.

\bibitem{dxs}
C. Ding, G. Xiao, W. Shan, The Stability Theory of Stream Ciphers,
Lecture Notes in Computer Science, Vol. 561, Springer, Berlin, 1991.

\bibitem{fd} X. Feng, Z. Dai, Expected value of the linear
complexity of two-dimensional binary sequences, in: T. Helleseth et
al. (Eds.), Sequences and Their Applications --- SETA 2004, Lecture
Notes in Computer Science, Vol. 3486, Springer, Berlin, 2005, pp.
113--128.

\bibitem{fwd} X. Feng, Q. Wang, Z. Dai, Multi-sequences with
$d$-perfect property, J. Complexity 21 (2005) 230--242.

\bibitem{fns} F.-W. Fu, H. Niederreiter, M. Su, The expectation and variance of
the joint linear complexity of random periodic multisequences, J.
Complexity 21 (2005) 804--822.

\bibitem{fno} F.-W. Fu, H. Niederreiter, F. \"Obudak, Joint linear complexity of
multisequences consisting of linear recurring sequences, in:
Cryptography and Communications-Discrete Structures, Boolean
Functions and Sequences, in press, available online at
doi:10.1007/s12095-007-0001-4.

\bibitem{fno1} F.-W. Fu, H. Niederreiter, F. \"Obudak,
Joint linear complexity of arbitrary multisequences consisting of
linear recurring sequences, Finite Fields Appl., accepted for
publication.

\bibitem{hr} P. Hawkes, G.G. Rose, Exploiting multiples of the connection polynomial
in word-oriented stream ciphers, in: T. Okamoto (Ed.), Advances in
Cryptology --- ASIACRYPT 2000, Lecture Notes in Computer Science,
Vol. 1976, Springer, Berlin, 2000, pp. 303--316.

\bibitem{kl} A. Klapper, Linear complexity of sequences under different
interpretations, IEICE Transactions on Fundamentals of Electronics,
Communications and Computer Sciences E89-A (2006) 2254-2257.

\bibitem{ln} R. Lidl, H. Niederreiter, Finite Fields, Addison-Wesley
Publishing Company, Massachusetts, 1983.

\bibitem{mei} W. Meidl, Discrete Fourier transform, joint linear complexity and
generalized joint linear complexity of multisequences, in: T.
Helleseth et al. (Eds.), Sequences and Their Applications --- SETA
2004, Lecture Notes in Computer Science, Vol. 3486, Springer,
Berlin, 2005, pp. 101--112.

\bibitem{mn1} W. Meidl, H. Niederreiter, The expected value of the joint linear complexity
 of periodic multisequences, J. Complexity 19 (2003) 61--72.

\bibitem{mnv} W. Meidl, H. Niederreiter, A. Venkateswarlu, Error linear
complexity measures for multisequences, J. Complexity 23 (2007) 169--192.

\bibitem{mo1}
W. Meidl, F. \"Ozbudak, Generalized joint linear complexity of
linear recurring multisequences. In: S.W. Golomb et al. (Eds.),
Sequences and Their Applications -- SETA 2008, Lecture Notes in
Computer Science, Vol. 5203, Springer, Berlin, 2008, pp. 266--277.

\bibitem{mo2}
W. Meidl, F. \"Ozbudak, Linear complexity over $\mathbb{F}_{q}$ and
over $\mathbb{F}_{q^m}$ for linear recurring sequences, Finite
Fields Appl. 15 (2009) 110--124.

\bibitem{mw} W. Meidl, A. Winterhof, On the joint linear complexity profile of
explicit inversive multisequences, J. Complexity 21 (2005)
324--336.

\bibitem{mus} G.L. Mullen, I. Shparlinski, Values of linear recurring sequences of
vectors over finite fields, Acta Arith. 65 (1993), 221-226.

\bibitem{n1} H. Niederreiter, Linear complexity and
related complexity measures for sequences, in: T. Johansson, S.
Maitra (Eds.), Progress in Cryptology --- INDOCRYPT 2003, Lecture
Notes in Computer Science, Vol. 2904, Springer, Berlin, 2003, pp.
1--17.

\bibitem{n2} H. Niederreiter, The probabilistic theory of the joint linear
complexity of multisequences, in: G. Gong et al. (Eds.), Sequences and Their
Applications --- SETA 2006, Lecture Notes in Computer Science, Vol. 4086,
Springer, Berlin, 2006, pp. 5--16.

\bibitem{nv} H. Niederreiter, A. Venkateswarlu, Periodic multisequences with large error linear
complexity, Des. Codes Cryptogr. 49 (2008) 33--45.

\bibitem{nw1} H. Niederreiter, L.P. Wang, Proof of a conjecture on the joint
linear complexity profile of multisequences, in: S.
Maitra et al. (Eds.), Progress in Cryptology --- INDOCRYPT 2005,
Lecture Notes in Computer Science, Vol. 3797, Springer, Berlin,
2005, pp. 13--22.

\bibitem{nw2} H. Niederreiter, L.P. Wang, The asymptotic behavior of the
joint linear complexity profile of multisequences, Monatsh. Math.
150 (2007) 141--155.

\bibitem{ru1} R.A. Rueppel, Analysis and Design of Stream Ciphers,
Springer, Berlin, 1986.

\bibitem{ru2} R.A. Rueppel, Stream ciphers, in: G.J. Simmons (Ed.),
Contemporary Cryptology: The Science of Information Integrity, IEEE
Press, New York, 1992, pp. 65--134.

\bibitem{wn} L.-P. Wang, H. Niederreiter, Enumeration results on the joint
linear complexity of multisequences, Finite Fields Appl. 12 (2006) 613--637.

\bibitem{wn1} L.-P. Wang, H. Niederreiter, Successive minima profile, lattice profile, and joint linear
complexity profile of pseudorandom multisequences, J. Complexity 24
(2008) 144--153.

\bibitem{xi} C.P. Xing, Multi-sequences with almost perfect linear complexity profile
and function fields over finite fields, J. Complexity 16 (2000)
661--675.

\end{thebibliography}
\end{document}